\documentclass[a4paper]{scrartcl}

\pdfoutput=1
\usepackage[english]{babel}
\usepackage{authblk}

\usepackage{graphicx}
\usepackage{latexsym}
\usepackage{textcomp}
\usepackage{stmaryrd}
\usepackage[nointegrals]{wasysym}
\usepackage{xspace}
\usepackage{enumitem}
\usepackage{subfig}
\usepackage{amsmath}
\usepackage{amssymb}
\usepackage{tikz}
\usetikzlibrary{arrows,shapes,snakes,automata,backgrounds,petri,patterns,spy}

\newcommand{\A}{\ensuremath{\mathcal{A}}\xspace}
\newcommand{\B}{\ensuremath{\mathcal{B}}\xspace}
\newcommand{\myqed}{}

\newcommand{\trans}[2][{}]{\mathbin{\smash{\overset{{#2}}{{\to}}_{#1}}}}
\newcommand{\btrans}[2][{}]{\,{\smash{\overset{{#2}}{{\nrightarrow}}_{#1}}}}

\usepackage{amsthm}
\newtheorem{definition}{Definition}

\newtheorem{theorem}{Theorem}

\begin{document}

%%%%%%%%%%%%%%%%%%%%%%%%%%%%%%%%%%%%%%%%%%%%%%%%%%%%%%%%%%%%%%%%%%%%%%
%%%%%%%%%%%%%%%%%%%%%%%%%%%%%%%%%%%%%%%%%%%%%%%%%%%%%%%%%%%%%%%%%%%%%%

\title{On the Complexity of Flanked\\ Finite State Automata}

\author[1,2,3]{Florent Avellaneda}
\author[1,3]{Silvano Dal Zilio}
\author[2,3]{Jean-Baptiste Raclet}
\affil[1]{CNRS, LAAS, F-31400 Toulouse,  France}
\affil[2]{IRIT, F-31400 Toulouse, France}
\affil[3]{Univ de Toulouse, F-31400 Toulouse, France}
\date{}
\maketitle
\begin{abstract}
  We define a new subclass of nondeterministic finite automata for
  prefix-closed languages called \emph{Flanked Finite Automata}
  (FFA). We show that this class enjoys good complexity properties
  while preserving the succinctness of nondeterministic automata. In
  particular, we show that the universality problem for FFA is in
  linear time and that language inclusion can be checked in polynomial
  time. A useful application of FFA is to provide an efficient way to
  compute the quotient and inclusion of regular languages without the
  need to use the powerset construction. These operations are the
  building blocks of several verification algorithms.
\end{abstract}

\section{Introduction}

While the problems of checking universality or language inclusion are
known to be computationally easy for Deterministic Finite Automata
(DFA), they are PSPACE-complete for Nondeterministic Finite Automata
(NFA). On the other hand, the size of a NFA can be exponentially
smaller than the size of an equivalent minimal DFA. This gap in
complexity between the two models can be problematic in practice. This
is for example the case when using finite state automata for system
verification, where we need to manipulate very large number of states.

Several work have addressed this problem by trying to find classes of
finite automata that retain the same complexity than DFA on some
operations while still being more succinct than the minimal DFA. A
good survey on the notion of determinism for automata is for
example~\cite{colcombet2012}. One such example is the
class of Unambiguous Finite Automata
(UFA)~\cite{schmidt78,stearns1985equivalence}. Informally, a UFA is a
finite state automaton such that, if a word is accepted, then there is
a unique run which witnesses this fact, that is a unique sequence of
states visited when accepting the word. Like with DFA, the problems of
universality and inclusion for UFA is in polynomial-time.

In this paper, we restrict our study to automaton accepting prefix
closed languages. More precisely, we assume that all the states of the
automaton are final (which corresponds exactly to the class of
prefix-closed regular languages). This restriction is very common when
using NFA for the purpose of system verification. For instance, Kripke
structures used in model-checking algorithms are often interpreted as
finite state automaton where all states are final. It is easy to see
that, with this restriction on prefix-closed language, an UFA is
necessarily deterministic. Therefore new classes of NFA, with the same
nice complexity properties than UFA, are needed in this context. We
can also note that the classical complexity results on NFA are still
valid when we restrict to automata accepting prefix-closed
language. For instance, given a NFA \A with all its states final,
checking the universality of \A is
PSPACE-hard~\cite{kao09:_nfas}. Likewise for the minimization
problem. Indeed there are examples of NFA with $n$ states, all finals,
such that the minimal equivalent DFA has $2^n$
states~\cite[Sect.~7]{kao09:_nfas}. We provide such an example in
Sect.~\ref{sec:succ-flank-autom} of this paper. Therefore this
restriction does not intrinsically change the difficulty of our task.

We define a new class of finite state automaton called \emph{Flanked
  Finite Automata} (FFA) that has complexity properties similar to
that of UFA but for prefix-closed language. Informally, a FFA includes
extra-information that can be used to check efficiently if a word is
not accepted by the automaton. In Sect.~\ref{sec:compl-results-basic},
we show that the universality problem for FFA is in linear-time while
testing the language inclusion between two FFA $\A$ and $\B$ is in
time $O(|\A|.|\B|)$, where $|\A|$ denotes the size of the automaton
$\A$ in number of states. In Sect.~\ref{prop:intersection}, we define
several operations on FFA. In particular we describe how to compute a
flanked automata for the intersection, union and quotient of two
languages defined by FFA. Finally, before concluding, we give an
example of (a family of) regular languages that can be accepted by FFA
which are exponentially more succinct than their equivalent minimal
DFA.

Our main motivation for introducing this new class of NFA was to
provide an efficient way to compute the quotient of two regular
languages $L_1$ and $L_2$. This operation, denoted $L_1 / L_2$ and defined in
Sect.~\ref{sec:clos-prop-flank}, is central in several automata-based
verification problems that arise in applications ranging from the
synthesis of discrete controller to the modular verification of
component-based systems. For example, it has been used in the
definition of contract-based specification
theories~\cite{benveniste2008multiple,bauer12} or as a key operation
for solving language equations~\cite{7202840}. With our approach, it
is possible to construct the quotient of two flanked automaton, $\A_1$
and $\A_2$, using less than $|\A_1|.|\A_2| + 1$ states; moreover the
resulting automata is still flanked. We believe that this work
provides the first algorithm for computing the quotient of two
languages without resorting to the powerset construction on the
underlying automata, that is without determinizing them.

\section{Notations and Definitions}

A finite automaton is a quintuple $\mathcal{A} = (Q, \Sigma, E, I)$
where: $Q$ is a finite set of states; $\Sigma$ is the alphabet of \A
(that is a finite set of symbols);
$E \subseteq Q \times \Sigma \times Q$ is the transition relation; and
$I \subseteq Q$ is the set of initial states. In the remainder of this
text, we always assume that every states of an automaton is final,
hence we do not need a distinguished subset of accepting
states. Without loss of generality, we also assume that every state in
$Q$ is reachable in \A from $I$ following a sequence of transitions in
$E$.

For every word $u \in \Sigma^*$ we denote $\mathcal{A}(u)$ the subset
of states in $Q$ that can be reached when trying to accept the word
$u$ from an initial state in the automaton. We can define the set
$\A(u)$ by induction on the word $u$. We assume that $\epsilon$ is the
empty word and we use the notation $u\,a$ for the word obtained form
$u$ by concatenating the symbol $a \in \Sigma$; then:
\[
\begin{array}{lcl}
  \A(\epsilon) & = & I\\
  \A(u\,a) & = & \{ q' \in Q \mid \exists q \in \A(u) . (q, a, q') \in E\}\\
\end{array}
\]

By extension, we say that a word $u$ is accepted by \A, denoted
$u \in \A$, if the set $\A(u)$ is not empty.

\begin{definition}
  A Flanked Finite Automaton (FFA) is a pair $(\mathcal{A}, F)$ where
  $\mathcal{A} = (Q, \Sigma, E, I)$ is a finite automaton and
  $F: Q \times \Sigma$ is a ``flanking function'', that associates
  symbols of $\Sigma$ to states of \A. We also require the following
  relation between \A and $F$:
  \begin{equation}\tag{\protect{F$\star$}} \label{eq:FF}
    \begin{array}{l}
      \forall u \in \Sigma^* , a \in \Sigma . \big ( \left
      ( u \in \A \wedge u\,a \notin \A \right ) \Leftrightarrow
      \exists q \in \mathcal{A}(u) . (q, a) \in F\, \big )
    \end{array}
  \end{equation}
  We will often use the notation $q \trans{a} q'$ when
  $(q, a, q') \in E$, that is when there is a transition from $q$ to
  $q'$ with symbol $a$ in \A. Likewise, we use the notation
  $q \btrans{a}$ when $(q, a) \in F$.
\end{definition}

With our condition that every state of an automaton is final, the
relation $q \trans{a} q'$ states that every word $u$ ``reaching'' $q$
in \A can be extended by the symbol $a$; meaning that $u\,a$ is also
accepted by \A. Conversely, the relation $q \btrans{a}$ states that
the word $u\,a$ is not accepted. Therefore, in a FFA $(\A, F)$, when
$q \in \A(u)$ and $(q, a) \in F$, then we know that the word $u$
cannot be extended with $a$. In other words, the flanking function
gives information on the ``frontier'' of a prefix-closed
language---the extreme limit over which words are no longer accepted
by the automaton---hence the use of the noun \emph{flank} to describe
this class.

In the rest of the paper, we simply say that the pair $(\A, F)$ is
\emph{flanked} when condition~\eqref{eq:FF} is met. We also say that
the automaton \A is \emph{flankable} if there exist a flanking
function $F$ such that $(\A, F)$ is flanked.

\subsection{Testing if a Pair $(\A, F)$ is Flanked}
\label{sec:testing-if-pair}

We can use the traditional Rabin-Scott powerset construction to test
whether $F$ flanks the automaton $\A = (Q, \Sigma, E, I)$.  We build
from \A the ``powerset automaton'' $\wp(\A)$, a DFA with alphabet
$\Sigma$ and with states in $2^Q$ (also called classes) that are the
sets of states in $Q$ reached after accepting a given word prefix;
that is all the sets of the form $\A(u)$. The initial state of
$\wp(\A)$ is the class $\A(\epsilon) = I$. Finally, we have that
$C \trans{a} C'$ in $\wp(\A)$ if and only if there is $q \in C$ and
$q' \in C'$ such that $q \trans{a} q'$.

Let $F^{-1}(a)$ be the set $\{ q \mid q \btrans{a} \}$ of states that
``forbids'' the symbol $a$ after a word accepted by \A.  Then the pair
$(\A, F)$ is flanked if, for every possible symbol $a \in \Sigma$ and
for every reachable class $C \in \wp(\A)$ we have:
$C \cap F^{-1}(a) \neq \emptyset$ if and only if there are no class
$C'$ such that $C \trans{a} C'$.

This construction shows that checking if a pair $(\A, F)$ is flanked
should be a costly operation, that is, it should be as complex as
exploring a deterministic automaton equivalent to \A. In
Sect.~\ref{sec:compl-results-basic} we prove that this problem is
actually PSPACE-complete.

\subsection{Testing if a NFA is Flankable}
\label{sec:testing-if-an}

It is easy to show that the class of FFA includes the class of
deterministic finite state automaton; meaning that every DFA is
flankable. If an automaton \A is deterministic, then it is enough to
choose the ``flanking function'' $F$ such that, for every state $q$ in
$Q$, we have $q \btrans{a}$ if and only if there are no transitions of
the form $q \trans{a} q'$ in \A. DFA are a proper subset of FFA;
indeed we give examples of NFA that are flankable in
Sect.~\ref{sec:succ-flank-autom}.

On the other way, if an automaton is not deterministic, then in some
cases it is not possible to define a suitable flanking function $F$.
For example, consider the automaton from Fig.~\ref{fig:flankable} and
assume, by contradiction, that we can define a flankable function $F$
for this automaton. The word $b$ is accepted by $\A$ but the word
$b\, b$ is not, so by definition of FFA (see eq.~\eqref{eq:FF}), there
must be a state $q \in \A(b)$ such that $q \btrans{b}$. Hence, because
$q_1$ is the only state in $\A(b)$, we should necessarily have
$q_1 \btrans{b}$.  However, this contradicts the fact that the word
$a\,b$ is in \A, since $q_1$ is also in $\A(a)$.

\tikzstyle{edge} = [draw,thick,->]
\begin{figure}
\centering
\begin{minipage}[c]{0.4\linewidth}
  \center
  \begin{tikzpicture}[scale=1.5]
    \node [place, double] (p0) at (0,0) {$q_0$};
    \node [place, double] (p1) at (-0.5,-1) {$q_1$};
    \node [place, double] (p2) at (0.5,-1) {$q_2$};
    \node [place, double] (p3) at (0.5,-2) {$q_3$};

    \path[edge] (p0) edge[bend left] node[right] {$a$} (p1);
    \path[edge] (p0) edge[bend right] node[left] {$b$} (p1);
    \path[edge] (p0) edge[bend left] node[right] {$a$} (p2);
    \path[edge] (p2) edge node[right] {$b$} (p3);

    \path[edge] node[above]  {} ++(-0.5,0) -- (p0) ;
  \end{tikzpicture}
\end{minipage}\hfill
\begin{minipage}[c]{0.4\linewidth}
  \center
  \begin{tikzpicture}[scale=1.5]
    \node (p0) at (0,0) {$\{q_0\}$};
    \node (p1) at (1,0.5) {$\{q_1\}$};
    \node (p12) at (1,-0.5) {$\{q_1, q_2\}$};
    \node (p3) at (2,-0.5) {$\{q_3\}$};

    \path[edge] (p0) edge[bend left] node[above] {$b$} (p1);
    \path[edge] (p0) edge[bend right] node[above] {$a$} (p12);
    \path[edge] (p12) edge node[above] {$b$} (p3);
  \end{tikzpicture}
\end{minipage}
\caption{\label{fig:flankable}An example of non-flankable NFA (left)
  and its associated Rabin-Scott powerset construction (right).}
\end{figure}
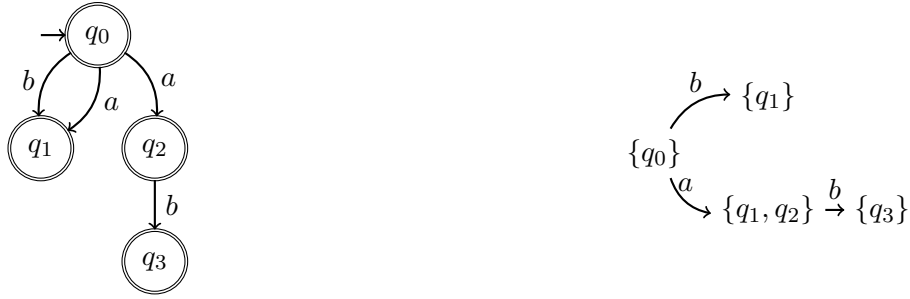

More generally, it is possible to define a necessary and sufficient
condition for the existence of a flanking function; this leads to an
algorithm for testing if an automaton \A is flankable. Let
$\A^{-1}(a)$ denotes the set of states reachable by words that can be
extended by the symbol $a$ (remember that we consider prefix-closed
languages):
\[
\A^{-1}(a) \ = \ \bigcup \{ \A(u) \mid u\,a \in \A \}
\]

It is possible to find a flanking function $F$ for the automaton $\A$
if and only if, for every word $u \in \A$ such that $u\, a \notin \A$
then the set $\A(u) \setminus \A^{-1}(a)$ is not empty. Indeed, in
this case, it is possible to choose $F$ such that $(q, a) \in F$ as
soon as there exists a word $u$ with
$q \in \A(u) \setminus \A^{-1}(a)$. 

Conversely, an automaton \A is not flankable if we can find a word
$u \in \A$ such that $u\,a \notin \A$ and
$\A(u) \subseteq \A^{-1}(a)$. For example, for the automaton in
Fig.~\ref{fig:flankable}, we have $\A^{-1}(b) = \{ q_0, q_1, q_2 \}$
while $b\,b \notin \A$ and $\A(b) = \{ q_1 \}$.

This condition can also be checked using the powerset
construction. Indeed, we can compute the set $\A^{-1}(a)$ by taking
the union of the classes in the powerset automaton $\wp(\A)$ that are
the source of an $a$ transition. Then it is enough to test this set
for inclusion against all the classes that have no outgoing
transitions labeled with $a$ in $\wp(\A)$.

\section{Complexity Results for Basic Problems}
\label{sec:compl-results-basic}

In this section we give some results on the complexity of basic
operations over FFA.

\begin{theorem}\label{prop:universality FFA}
  The universality problem for FFA is decidable in linear time.
\end{theorem}

\begin{proof}
  We consider a FFA $(\A, F)$ with $\A = (Q, \Sigma, E, I)$ and we
  want to check that every word $u \in \Sigma^*$ is accepted by \A. We
  assume that $Q$ and $I$ are not empty and that every state is
  reachable in \A. We also assume that the function $F$ is ``encoded''
  a mapping from $Q$ to sequences of symbols in $\Sigma$.

  We start by proving that \A is universal if and only if the relation
  $F$ is empty; meaning that for all states $q \in Q$ it is not
  possible to find a symbol $a \in \Sigma$ such that $q \btrans{a}$.
  As a consequence, all words reaching a state $q$ in \A can always be
  extended by any symbol of $\Sigma$.

  \begin{description}
  \item[{\A universal implies $F$ empty.}] If the automaton \A is
    universal then every word $u \in \Sigma^*$ is accepted by \A and
    can be extended by any symbol $a \in \Sigma$.  Hence, by
    definition of FFA (see eq.~\eqref{eq:FF}) we have that
    $(q, a) \notin F$ for all symbol $a$ in $\Sigma$. Hence
    $F$ is the empty relation over $Q \times \Sigma$.\\

  \item[{\A not universal implies $F$ not empty.}] Assume that $u$ is
    the shortest word not accepted by \A.  We have that
    $u \neq \epsilon$, since $I$ is not empty. Hence there exist a
    word $v$ such that $u = v \, a$ and $v$ is accepted. Again, by
    definition of FFA (see eq.~\eqref{eq:FF}), there must be a state
    $q \in \A(v)$ such that $q \btrans{a}$; and therefore $F$ is not
    empty.
  \end{description}

  As a consequence, to test whether \A is universal, it is enough to
  check whether there is a state $q \in Q$ that is mapped to a
  non-empty set of symbols in $F$. Note that, given a different
  encoding of $F$, this operation could be performed in constant
  time.\myqed
\end{proof}

We can use this result to settle the complexity of testing if an
automaton is flankable.

\begin{theorem}
  Given an automaton $\A = (Q, \Sigma, E, I)$ and a relation
  $F \in Q \times \Sigma$, the problem of testing if $(\A, F)$ is a
  flanked automaton is PSPACE-complete when there is at least two
  symbols in $\Sigma$.
\end{theorem}

\begin{proof}
  We can define a simple nondeterministic algorithm for testing is
  $(\A, F)$ is flanked. We recall that the function $F^{-1}(a)$ stands
  for the set $\{ q \mid q \btrans{a} \}$ of states that ``forbids''
  the symbol $a$. As stated in Sect.~\ref{sec:testing-if-pair}, to
  test if $(\A, F)$ is flanked, we need, for every symbol
  $a \in \Sigma$, to explore the classes $C$ in the powerset automaton
  of \A and test whether $C \trans{a} C'$ in $\wp(\A)$ and whether
  $C \cap F^{-1}(A) = \emptyset$ or not. These tests can be performed
  using $|Q|$ bits since every class $C$ and every set $F^{-1}(a)$ is
  a subset of $Q$. Moreover there are at most $2^{|Q|}$ classes in
  $\wp(\A)$. Hence, using Savitch's theorem, the problem is in PSPACE.

  On the other way, we can reduce the problem of testing the
  universality of a NFA $\A$ to the problem of testing if a pair
  $(\A, \emptyset)$, where $\emptyset$ is the ``empty'' flanking
  function over $Q \times \Sigma$. The universality problem is known
  to be PSPACE-hard when the alphabet $\Sigma$ is of size at least
  $2$, even if all the states of \A are
  final~\cite{kao09:_nfas}. Indeed, to test if \A is universal, we
  showed in the proof of the previous theorem, that it is enough to
  check that $(\A, \emptyset)$ is flanked. Hence our problem is also
  PSPACE-hard.\myqed
\end{proof}

To conclude this section, we prove that the complexity of checking
language inclusion between a NFA and a FFA is in polynomial time,
therefore proving that our new class of automata as the same nice
complexity properties than those of UFA. We say that the language of
$\A_1$ is included in $\A_2$, simply denoted $\A_1 \subseteq \A_2$, if
all the words accepted by $\A_1$ are also accepted by $\A_2$.

\begin{theorem}
  Given a NFA $\A_1$ and a FFA $(\A_2, F_2)$, we can check whether
  $\A_1 \subseteq \A_2$ in polynomial time.
\end{theorem}

\begin{proof}
  Without loss of generality, we can assume that
  $\A_1 = (Q_1, \Sigma, E_1, I_1)$ and
  $\A_2 = (Q_2, \Sigma, E_2, I_2)$ are two NFA over the same alphabet
  $\Sigma$. We define a variant of the classical product construction
  between $\A_1$ and $\A_2$ that also takes into account the
  ``pseudo-transitions'' $q \btrans{a}$ defined by the flanking
  functions.
  
  We define the product of $\A_1$ and $(\A_2, F_2)$ as the NFA
  $\A = (Q, \Sigma, E, I)$ such that $I = I_1 \times I_2$ and
  $Q = (Q_1 \times Q_2) \cup \{ \bot \}$. The extra state $\bot$ will
  be used to detect an ``error condition'', that is a word that is
  accepted by $\A_1$ and not by $\A_2$. The transition relation of \A
  is such that:
  \begin{itemize}
  \item if $q_1 \trans{a} q'_1$ in $\A_1$ and $q_2 \trans{a} q'_2$ in
    $\A_2$ then $(q_1, q_2) \trans{a} (q'_1, q'_2)$ in \A;
  \item if $q_1 \trans{a} q'_1$ in $\A_1$ and $q_2 \btrans{a}$ in
    $\A_2$ then $(q_1, q_2) \trans{a} \bot$ in \A
  \end{itemize}

  We can show that the language of $\A_1$ is included in the language
  of $\A_2$ if and only if the state $\bot$ is not reachable in
  \A. Actually, we show that any word $u$ such that $\bot \in \A(u)$
  is a word accepted by $\A_1$ and not by $\A_2$.

  We prove the first implication. Assume that every word $u$ accepted
  by $\A_1$ is accepted by $\A_2$. Hence we can prove by induction on
  the size of $u$ that $\A(u) \subseteq Q_1 \times Q_2$. On the other
  way, if $u$ is not accepted by $\A_1$ then $u$ is not accepted by
  $\A$ (there are no transitions in this case). Hence, for all words
  in $\Sigma^*$, the set $\A(u)$ does not contain $\bot$.

  For the other direction, assume that there is a word $u$ such that
  $\bot \in \A(u)$. The word $u$ cannot be $\epsilon$ since
  $\A(\epsilon) = I_1 \times I_2 \not\ni \bot$. Therefore $u$ is of
  the form $v\, a$.  Since there are no transitions from $\bot$ in \A,
  there must be a pair $(q_1, q_2) \in Q_1 \times Q_2$ such that
  $q_1 \in \A_1(v)$; $q_2 \in \A_2(v)$; $q_1 \trans{a} q'_1$ in $\A_1$
  and $q_2 \btrans{a}$ in $\A_2$. By property~\eqref{eq:FF}, since
  $(\A_1, F_1)$ and $(\A_2, F_2)$ are both flanked, we have that
  $v\, a \in \A_1$ and $v\, a \notin \A_2$, as required.

  We cannot generate more than $|Q_1|.|Q_2|$ reachable states in \A
  before finding the error $\bot$ (or stopping the
  construction). Hence this algorithm is solvable in polynomial
  time.\myqed
\end{proof}

\section{Closure Properties of Flanked Automata}
\label{sec:clos-prop-flank}

In this section, we study how to compute the composition of flanked
automata. We prove that the class of FFA is closed by language
intersection and by the ``intersection adjunct'', also called
quotient. On a negative side, we show that the class is not closed by
non-injective relabeling.

We consider the problem of computing a flanked automaton accepting the
intersection of two prefix-closed, regular languages. More precisely,
given two FFA $(\A_1, F_1)$ and $(\A_2, F_2)$, we want to compute a
FFA $(\A, F)$ that recognizes the set of words accepted by both $\A_1$
and $\A_2$, denoted simply $\A_1 \cap \A_2$.

\begin{theorem}\label{prop:intersection}
  Given two FFA $(\A_1, F_1)$ and $(\A_2, F_2)$, we can compute a FFA
  $(\mathcal{A}, F)$ for the language $\A_1 \cap \A_2$ in polynomial
  time. The NFA $\A$ has size less than $|\A_1|.|\A_2|$.
\end{theorem}

\begin{proof}
  We define a classical product construction between $\A_1$ and $\A_2$
  and show how to extend this composition on the flanking functions.
  We assume that $\A_i$ is an automaton $(Q_i, \Sigma, E_i, I_i)$ for
  $i \in \{1,2\}$.

  The automaton $\mathcal{A} = (Q, \Sigma, E, I)$ is defined as the
  synchronous product of $\mathcal{A}_1$ and $\mathcal{A}_2$, that is:
  $Q = Q_1 \times Q_2$; $I = I_1 \times I_2$; and the transition
  relation is such that $(q_1, q_2) \trans{a} (q'_1, q'_2)$ in \A if
  both $q_1 \trans{a} q'_1$ in $\A_1$ and $q_2 \trans{a} q'_2$ in
  $\A_2$. It is a standard result that $\A$ accepts the language
  $\A_1 \cap \A_2$.

  The \emph{flanking function} $F$ is defined as follows: for each
  accessible state $(q_1, q_2) \in Q$, we have $(q_1, q_2) \btrans{a}$
  if and only if $q_1 \btrans{a}$ in $\A_1$ or $q_2 \btrans{a}$ in
  $\A_2$. What is left to prove is that $(\A, F)$ is flanked, that is,
  we show that condition~\eqref{eq:FF} is correct:
  \begin{itemize}
  \item assume $u$ is accepted by $\A$ and $u\, a$ is not; then there
    is a state $q = (q_1, q_2)$ in \A such that $q \in \A(u)$ and
    $(q, a) \in F$. By definition of \A, we have that $u$ is accepted
    by both $\A_1$ or $\A_2$, while the word $u\, a$ is not accepted
    by at least one of them. Assume that $u\, a$ is not accepted by
    $\A_1$. Since $F_1$ is a flanking function for $\A_1$, we have by
    equation~\eqref{eq:FF} that $(q_1, a) \in F_1$; and therefore $(q,
    a) \in F$, as required.

  \item assume there is a reachable state $q = (q_1, q_2)$ in \A such
    that $q \in \A(u)$ and $(q, a) \in F$; then $u$ is accepted by
    \A. We show, by contradiction, that $u\,a$ cannot be accepted by
    \A, that is $u \, a \notin \A_1 \cap \A_2$. Indeed, if so, then
    $u\,a$ will be accepted both by $\A_1$ and $\A_2$ and therefore we
    will have $(q_1, a) \notin F_1$ and $(q_2, a) \notin F_2$, which
    contradicts the fact that $(q, a) \in F$.
    \end{itemize}
\end{proof}

Next we consider the adjunct of the intersection operation, denoted
$\A_1 / \A_2$. This operation, also called \emph{quotient}, is defined
as the biggest prefix-closed language $X$ such that
$\A_2 \cap X \subseteq \A_1$. Informally, $X$ is the solution to the
following question: what is the biggest set of words $x$ such that $x$
is either accepted by $\A_1$ or not accepted by $\A_2$. Therefore the
language $\A_1 / \A_2$ is always defined (and not empty), since it
contains at least the empty word $\epsilon$. Actually, the quotient
can be interpreted as the biggest prefix-closed language included in
the set $L_1 \cup \bar{L_2}$, where $L_1$ is the language accepted by
$\A_1$ and $\bar{L_2}$ is the complement of the language of
$\A_2$. The quotient operation can also be defined by the following
two axioms:
\[ 
\text{(Ax1)} \quad \A_2 \cap (\A_1 / \A_2) \subseteq \A_1
\qquad\qquad
\text{(Ax2)}\quad 
\forall X .\, \A_2 \cap X
\subseteq \A_1 \Rightarrow X \subseteq \A_1 / \A_2
\]

The quotient operation is useful when trying to solve \emph{language
  equations problems}~\cite{7202840} and has applications in the
domain of system verification and synthesis. For instance, we can find
a similar operation in the contract framework of Benveniste et
al.~\cite{benveniste2008multiple} or in the contract framework of
Bauer et al.~\cite{bauer12}.

Our results on FFA can be use for the simplest instantiation of these
frameworks, that considers a simple trace-based semantics where the
behavior of systems is given as a regular set of words; composition is
language intersection; and implementation refinement is language
inclusion. Our work was motivated by the fact that there are no known
effective methods to compute the quotient. Indeed, to the best of our
knowledge, all the approaches rely on the determinization of NFA,
which is very expensive in practice~\cite{Raclet200893,7202840}.

Our definitions of quotient could be easily extended to replace
language intersection by synchronous product and to take into account
the addition of modalities~\cite{Raclet200893}.

\begin{theorem}\label{prop:quotient}
  Given two FFA $(\A_1, F_1)$ and $(\A_2, F_2)$, we can compute a FFA
  $(\mathcal{A}, F)$ for the quotient language $\A_1 / \A_2$ in
  polynomial time. The NFA $\A$ has size less than $|\A_1|.|\A_2| + 1$
\end{theorem}

\begin{proof}
  Without loss of generality, we can assume that
  $\A_1 = (Q_1, \Sigma, E_1, I_1)$ and
  $\A_2 = (Q_2, \Sigma, E_2, I_2)$ are two NFA over the same alphabet
  $\Sigma$. Like in the construction for testing language inclusion,
  we define a variant of the classical product construction between
  $\A_1$ and $\A_2$ that also takes into account the flanking
  functions.
  
  We define the product of $(\A_1, F_1)$ and $(\A_2, F_2)$ as the NFA
  $\A = (Q, \Sigma, E, I)$ such that $I = I_1 \times I_2$ and
  $Q = (Q_1 \times Q_2) \cup \{ \top \}$. The extra state $\top$ will
  be used as a sink state from which every suffix can be accepted. The
  transition relation of \A is such that:
  \begin{itemize}
  \item if $q_1 \trans{a} q'_1$ in $\A_1$ and $q_2 \trans{a} q'_2$ in
    $\A_2$ then $(q_1, q_2) \trans{a} (q'_1, q'_2)$ in \A;
  \item if $q_2 \btrans{a}$ in $\A_2$ then $(q_1, q_2) \trans{a} \top$
    in \A for all state $q_1 \in Q_1$
  \item $\top \trans{a} \top$ for every $a \in \Sigma$
  \end{itemize}

  Note that we do not have a transition rule for the case where
  $q_1 \btrans{a}$ in $\A_1$ and $q_2 \trans{a} q'_2$; this models the
  fact that a word ``that can be extended'' in $\A_2$ but not in
  $\A_1$ cannot be in the quotient $\A_1 / \A_2$. It is not difficult
  to show that $\A$ accepts the language $\A_1 / \A_2$. We give an
  example of the construction in Figure~\ref{fig:quotien}.

  Next we show that \A is flankable and define a suitable flanking
  function. Let $F$ be the relation in $Q \times \Sigma$ such that
  $(q_1, q_2) \btrans{a}$ if and only if $q_1 \btrans{a}$ in $F_1$ and
  $q_2 \trans{a} q'_2$ in $\A_2$. That is, the symbol $a$ is forbidden
  exactly in the case that was ruled out in the transition relation of
  \A. What is left to prove is that $(\A, F)$ is flanked, that is, we
  show that condition~\eqref{eq:FF} is correct:
  \begin{itemize}
  \item Assume $u$ is accepted by $\A$ and $u\, a$ is not.  Since
    $u\,a$ is not accepted, it must be the case that $q \neq \top$.
    Therefore there is a state $q = (q_1, q_2)$ in \A such that
    $q_1 \in \A_1(u)$ and $q_2 \in \A_2(u)$.  Also, since there are no
    transition with label $a$ from $q$, then necessarily
    $q_1 \btrans{a}$ in $\A_1$ and $q_2 \trans{a} q'_2$. This is
    exactly the case where $(q, a) \in F$, as required.

  \item Assume there is a reachable state $q$ in \A such that
    $q \in \A(u)$ and $(q, a) \in F$.  Since $(q, a) \in F$, we have
    $q \neq \top$ and therefore $q = (q_1, q_2)$ with
    $q_1 \in \A_1(u)$ and $q_2 \in \A_2(u)$. Hence $u$ is accepted by
    \A. Next, we show by contradiction that $u\,a$ cannot be accepted
    by \A. Indeed, if it was the case then $u\,a \in \A_2$ and
    $u\,a \notin \A_2$.  However, if $u\,a \in \A_2$ then,
    $(q_2, a) \notin F_2$ and so, by construction,
    $((q_1, q_2), a) \notin F$.
    \end{itemize}
    \myqed
\end{proof}

We give an example of the construction of the ``quotient'' FFA in
Fig.~\ref{fig:quotien}. If we look more closely at the construction
used in Theorem~\ref{prop:quotient}, that defines an automaton for
the quotient of two FFA $(\A_1, F_1)$ and $(\A_2, F_2)$, we see that
the flanking function $F_1$ is used only to compute the flanking
function of the result. Therefore, as a corollary, it is not difficult
to prove that we can use the same construction to build a quotient
automaton for $\A_1 / \A_2$ from an arbitrary NFA $\A_1$ and a FFA
$(\A_2, F_2)$. However the resulting automaton may not be flankable.

\begin{figure}
\centering
\begin{minipage}[t]{0.3\linewidth}
  \centering
  \caption*{$F_1 = \{(q_0, b), (q_1, a)\}$\\$\mathcal{A}_1 :$}
  \begin{tikzpicture}[scale=1.5]
    \node [place, double] (p0) at (0,0) {$q_0$};
    \node [place, double] (p1) at (0,-1) {$q_1$};

    \path[edge] node[above]  {} ++(-0.5,0) -- (p0) ;

    \path[edge] (p0) edge[bend right] node[left] {$a$} (p1);
    \path[edge] (p1) edge[bend right] node[right] {$b$} (p0);
  \end{tikzpicture}
\end{minipage}\hfill
\begin{minipage}[t]{0.3\linewidth}
  \centering
  \caption*{$F_2 = \{(q_1, a)\}$\\$\mathcal{A}_2 :$}
  \begin{tikzpicture}[scale=1.5]
    \node [place, double] (p0) at (0,0) {$q_0$};
    \node [place, double] (p1) at (0,-1) {$q_1$};

    \path[edge] node[above]  {} ++(-0.5,0) -- (p0) ;

    \path[edge] (p0) edge node[left] {$a$} (p1);
    \path[edge] (p0) edge[loop right] node[right] {$b$} (p0);
    \path[edge] (p1) edge[loop right] node[right] {$b$} (p1);

  \end{tikzpicture}
\end{minipage}\hfill
\begin{minipage}[t]{0.4\linewidth}
  \centering
  \caption*{$F = \{(q_0, b), (q_2, b)\}$\\$\mathcal{A} = \A_1
    / \A_2 :$}
  \begin{tikzpicture}[scale=1.5]
    \node [place, double] (p0) at (0,0) {$\,q_0, q_0\,$};
    \node [place, double] (p1) at (0,-1) {$\,q_1, q_1\,$};
    \node [place, double] (p2) at (0,-2) {$\,q_0, q_1\,$};
    \node [place, double] (p3) at (1,-1.5) {$\top$};

    \path[edge] node[above]  {} ++(-0.5,0) -- (p0) ;

    \path[edge] (p0) edge node[left] {$a$} (p1);
    \path[edge] (p1) edge node[left] {$b$} (p2);
    \path[edge] (p1) edge node[above] {$a$} (p3);
    \path[edge] (p2) edge node[below] {$a$} (p3);

    \path[edge] (p3) edge[loop right] node[right] {$a, b$} (p3);

  \end{tikzpicture}
\end{minipage}
\caption{\label{fig:quotien} Construction for the quotient of two FFA
  $(\mathcal{A}_1, F_1)$ and $(\mathcal{A}_2, F_2)$.}
\end{figure}

We can also prove that flankability is preserved by language union:
given two FFA $(\A_1, F_1)$ and $(\A_2, F_2)$, we can compute a FFA
$(\A, F)$ that recognizes the set of words accepted either by $\A_1$
or by $\A_2$, denoted $\A_1 \cup \A_2$. Operations corresponding to
the adjunct of the union or the to Kleene star closure are not
interesting in the context of automaton where every state is final and
therefore they are not studied in this paper.

\begin{theorem}
  Given two FFA $(\A_1, F_1)$ and $(\A_2, F_2)$, we can compute a FFA
  $(\mathcal{A}, F)$ for the language $\A_1 \cup \A_2$ in polynomial
  time. The NFA $\A$ has size less than $(|\A_1| + 1).(|\A_2| + 1)$.
\end{theorem}

\begin{proof}
  Like for language intersection and language inclusion, we base our
  construction on a variant of the classical product construction
  between $\A_1$ and $\A_2$ and show how to extend this composition on
  the flanking functions. We assume that $\A_i$ is an automaton
  $(Q_i, \Sigma, E_i, I_i)$ for $i \in \{1,2\}$ and that both
  automaton have the same alphabet.

  We consider a special state symbol $\top$ not in $Q_1 \cup Q_2$.
  This state will be used in $\A$ when we start accepting words that
  are not in the intersection of $\A_1$ and $\A_2$. The automaton
  $\mathcal{A} = (Q, \Sigma, E, I)$ is such that:
  $Q \subseteq (Q_1 \cup \{\top\}) \times (Q_2 \cup \{ \top\})$;
  $I = I_1 \times I_2$; and the transition relation is such that:
  \begin{itemize}
  \item if $q_1 \trans{a} q'_1$ in $\A_1$ and $q_2 \trans{a} q'_2$ in
    $\A_2$ then $(q_1,q_2) \trans{a} (q'_1,q'_2)$ in $\A$;
  \item if $q_1 \trans{a} q'_1$ in $\A_1$ and $q_2 \btrans{a}$ in
    $\A_2$ then $(q_1,q_2) \trans{a} (q'_1,\top)$ in $\A$;
  \item if $q_1 \btrans{a}$ in $\A_1$ and $q_2 \trans{a} q'_2$ in
    $\A_2$ then $(q_1,q_2) \trans{a} (\top,q'_2)$ in $\A$;
  \item if $q_1 \trans{a} q'_1$ in $\A_1$ then
    $(q_1,\top) \trans{a} (q'_1,\top)$ in $\A$;
  \item if $q_2 \trans{a} q'_2$ in $\A_2$ then
    $(\top, q_2) \trans{a} (\top, q'_2)$ in $\A$.
  \end{itemize}
  It is not difficult to prove that the NFA $\A$ accepts all the words
  in $\A_1 \cup \A_2$.

  The \emph{flanking function} $F$ is defined as the smallest relation
  such that, for each accessible state
  $(q_1, q_2) \in Q_1 \times Q_2$:
  \begin{itemize}
  \item if both $q_1 \btrans{a}$ in $F_1$ and $q_2 \btrans{a}$ in
    $F_2$ then $(q_1, q_2) \btrans{a}$ in $F$;
  \item if  $q_1 \btrans{a}$ in $F_1$ then $(q_1, \top) \btrans{a}$ in $F$;
  \item if $q_2 \btrans{a}$ in $F_2$ then $(\top, q_2) \btrans{a}$ in
    $F$
  \end{itemize}

  We are left to prove that $(\A, F)$ is flanked, that is
  condition~\eqref{eq:FF} is correct. The proof is very similar to the
  one for Theorem~\ref{prop:intersection}.\myqed
\end{proof}

The two main closure properties given in this section are useful when
we want to check language inclusion between the composition of several
languages; for example if we need to solve, for $X$, the equation
$\A_1 \cap \dots \cap \A_n \cap X \subseteq \B$. This is the case, for
example, if we need to synthesize a discrete controller, $X$, that
satisfies a given requirement specification $\B$ when put in parallel
with components whose behavior is given by $\A_i$ (with $i \in 1..n$).
Indeed, even though there may be a small price to pay to ``flank'' the
sub-components of this equation, we can incrementally build a flanked
automaton for $\A_1 \cap \dots \cap \A_n$ and then compute efficiently
the quotient $\B / (\A_1 \cap \dots \cap \A_n)$.

Even though the class of FFA enjoys interesting closure properties,
there are operations that, when applied to a FFA, may produce a result
that is not flankable. This is for example the case with
``(non-injective) relabeling'', that is the operation of applying a
substitution over the symbols of an automaton. The same can be
observed if we consider an erasure operation, in which we can replace
all transition on a given symbol by an
$\epsilon$-transition. Informally, it appears that the property
flankable can be lost when applying an operation that increases the
non-determinism of the transition relation.

We can prove this result by exhibiting a simple counterexample, see
the automaton in Fig.~\ref{fig:relabeling}. This automaton with
alphabet $\Sigma = \{a, b, c\}$ is deterministic, so we can easily
define an associated flanking function. For example we can choose
$F = \{ (q_1, a), (q_1, b), (q_1, c), (q_2, a),$
$(q_2, c), (q_3, a), (q_3, b), (q_3, c)\}$.  However, if we substitute
the symbol $c$ with $a$ (that is we apply the non-injective relabeling
function $\{a \leftarrow a\}\{b \leftarrow b\}\{c \leftarrow a\}$), we
obtain the non-flankable automaton described in
Sect.~\ref{sec:testing-if-pair} (see Fig.~\ref{fig:flankable}).

\begin{figure}
\centering
  \begin{tikzpicture}[scale=1.5]
    \node [place, double] (p0) at (0,0) {$q_0$};
    \node [place, double] (p1) at (-0.5,-1) {$q_1$};
    \node [place, double] (p2) at (0.5,-1) {$q_2$};
    \node [place, double] (p3) at (0.5,-2) {$q_3$};

    \path[edge] (p0) edge[bend left] node[right] {$a$} (p1);
    \path[edge] (p0) edge[bend right] node[left] {$b$} (p1);
    \path[edge] (p0) edge[bend left] node[right] {$c$} (p2);
    \path[edge] (p2) edge node[right] {$b$} (p3);

    \path[edge] node[above]  {} ++(-0.5,0) -- (p0) ;
  \end{tikzpicture}
  \caption{\label{fig:relabeling}Example of a FFA not flankable after relabeling $c$ to $a$.}
\end{figure}
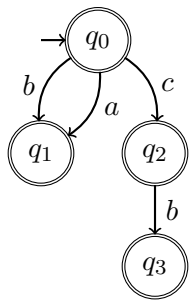

\section{Succinctness of Flanked Automata}
\label{sec:succ-flank-autom}

In this section we show that a flankable automata can be exponentially
more succinct than its equivalent minimal DFA. This is done by
defining a language over an alphabet of size $2\,n$ that can be
accepted by a linear size FFA but that corresponds to a minimal DFA
with an exponential number of states. This example is due to Thomas
Colcombet~\cite{colcombet15:_flank}.

At first sight, this result may seem quite counterintuitive. Indeed,
even if a flanked automata is build from a NFA, the combination of the
automaton and the flanking function contains enough information to
``encode'' both a language and its complement. This is what explain
the good complexity results on testing language inclusion for
example. Therefore we could expect worse results concerning the
relative size of a FFA and an equivalent DFA.

\begin{theorem}\label{thm:FFAvsDFA}
  For every integer $n$, we can find a FFA $(\A_n, F)$ such that
  $\A_n$ has $2\,n + 2$ states and that the language of $\A_n$ cannot
  be accepted by a DFA with less than $2^n$ states.
\end{theorem}

\begin{proof}
  We consider two alphabets with $n$ symbols:
  $\Pi_n = \{ 1, \dots, n \}$ and
  $\Theta_n = \{\sharp_1,\sharp_2,\dots,\sharp_n\}$. We define the
  language $L_n$ over the alphabet $\Pi_n \cup \Theta_n$ as the
  smallest set of words such that:
  \begin{itemize}
  \item all words in $\Pi_n^*$ are in $L_n$, that is all the words
    that do not contain a symbol of the kind $\sharp_i$;
  \item a word of the form $(u \, \sharp_i)$ is in $L_n$ if and only
    if $u$ is a word of $\Pi_n^*$ that contains at least one
    occurrence of the symbol $i$. That is $L_n$ contains all the words
    of the form $\Pi_n^* \cdot i \cdot \Pi_n^* \cdot \sharp_i$ for all
    $i \in 1..n$. We denote $L^i_n$ the regular language consisting of
    the words of the form
    $\Pi_n^* \cdot i \cdot \Pi_n^* \cdot \sharp_i$.
  \end{itemize}
  Clearly the language $L_n$ is the union of $n+1$ regular languages;
  $L = \Pi_n^* \cup L^1_n \cup \dots \cup L^n_n$. It is also easy to
  prove that $L_n$ is prefix-closed, since the set of prefixes of the
  words in $L^i_n$ is exactly $\Pi_n^*$ for all $i \in 1..n$.

  A DFA accepting the language $L_n$ must have at least $2^n$
  different states. Indeed it must be able to record the subset of
  symbols in $\Pi_n$ that have already been seen before accepting
  $\sharp_i$ as a final symbol; to accept a word of the form
  $u \, \sharp_i$ the DFA must know whether $i$ has been seen in $u$
  for all possible $i \in 1..n$.

  Next we define a flankable NFA
  $\A_n = (Q_n, \Pi_n \cup \Theta_n, E_n, \{p\})$ with $2\,n + 2$
  states that can recognize the language $L_n$. We give an example of
  the construction in Fig.~\ref{img:exampleLIAFA} for the case
  $n = 3$.  The NFA $\A_n$ has a single initial state, $p$, and a
  single sink state (a state without outgoing transitions), $r$. The
  set $Q_n$ also contains two states, $p_i$ and $q_i$, for every
  symbol $i$ in $\Pi$.

  The transition relation $E_n$ is the smallest relation that contains
  the following triplets for all $i \in 1..n$:
  \begin{itemize}
  \item the $3$ transitions $p \trans{i} q_i$;
    $p_i \trans{i} q_i$; and $q_i \trans{i} q_i$;
  \item for every index $j \neq i$, the $3$ transitions
    $p \trans{j} p_i$; $p_i \trans{j} p_i$; and $q_i \trans{j} q_i$;
  \item and the transition $q_i \trans{\sharp_i} r$.
  \end{itemize}

  Intuitively, a transition from $p$ to $p_i$ or $q_i$ will select
  non-deterministically which final symbol $\sharp_i$ is expected at
  the end of the word (which sub-language $L^i_n$ we try to
  accept). Once a symbol in $\Theta$ has been seen---in one of the
  transition of the kind $q_i \trans{\sharp_i} r$---the automaton is
  stuck on the state $r$. It is therefore easy to prove that $\A_n$
  accepts the union of the languages $L^i_n$ and their prefixes.

  Finally, the NFA $\A_n$ is flankable. It is enough to choose, for
  the flanking function, the smallest relation on $Q \times \Theta_n$
  such that $p_i \btrans{\sharp_i}$ and $p \btrans{\sharp_i}$ for all
  $i \in 1..n$; and such that $r \btrans{a}$ for all the symbols
  $a \in \Pi_n \cup \Theta_n$. Indeed, it is not possible to accept
  the symbol $\sharp_i$ from the initial state, $p$, or from a word
  that can reach $p_i$; that is, it is not possible to extend a word
  without any occurrence of the symbol $i$ with the symbol
  $\sharp_i$. Also, it is not possible to extend a word that can reach
  the state $r$ in $\A_n$. It is easy to prove that this cover all the
  possible words not accepted by $\A_n$.\myqed
\end{proof}

%%%%%%%%%%%%%%%%%%%%%%%%%%%%%%%%%%%%%%%%%%%%%%%%%%%%%%%%%%%%%%%%%%%%%%
%%%%%%%%%%%%%%%%%%%%%%%%%%%%%%%%%%%%%%%%%%%%%%%%%%%%%%%%%%%%%%%%%%%%%%

\begin{figure}[h!]
  \center
  \begin{tikzpicture}[scale=1.2]
    \node [place, double] (p0) at (0,0) {};

    \node [place, double] (pa) at (-2,0) {};
    \node [place, double] (pb) at (0,-1) {};
    \node [place, double] (pc) at (+2,0) {};

    \node [place, double] (pA) at (-2,-1.5) {};
    \node [place, double] (pB) at (0,-2) {};
    \node [place, double] (pC) at (+2,-1.5) {};

    \node [place, double] (pf) at (0,-3) {};

    \path[edge] node[above]  {} ++(0,+0.5) -- (p0) ;

    \path[edge] (p0) edge node[above] {$2, 3$} (pa);
    \path[edge] (p0) edge node[right] {$1, 3$} (pb);
    \path[edge] (p0) edge node[above] {$1, 2$} (pc);

    \path[edge] (p0) edge node[left] {$1$} (pA);
    \path[edge] (p0) edge[bend right] node[left] {$2$} (pB);
    \path[edge] (p0) edge node[right] {$3$} (pC);

    \path[edge] (pa) edge[loop left] node[above] {$2, 3$} (pa);
    \path[edge] (pb) edge[loop right] node[above] {$1, 3$} (pb);
    \path[edge] (pc) edge[loop right] node[above] {$1, 2$} (pc);

    \path[edge] (pa) edge node[left] {$1$} (pA);
    \path[edge] (pb) edge node[right] {$2$} (pB);
    \path[edge] (pc) edge node[right] {$3$} (pC);

    \path[edge] (pA) edge[loop left] node[above] {$1, 2, 3$} (pA);
    \path[edge] (pB) edge[loop right] node[above] {$1, 2, 3$} (pB);
    \path[edge] (pC) edge[loop right] node[above] {$1, 2, 3$} (pC);

    \path[edge] (pA) edge node[left] {$\sharp_1$} (pf);
    \path[edge] (pB) edge node[right] {$\sharp_2$} (pf);
    \path[edge] (pC) edge node[right] {$\sharp_3$} (pf);
  \end{tikzpicture}
    \caption{\label{img:exampleLIAFA}Flankable NFA for the language $L_3$.}
\end{figure}
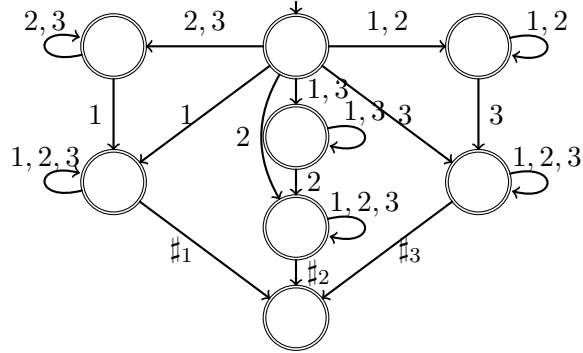

\section{Conclusion}
\label{sec:conclusion}

We define a new subclass of NFA for prefix-closed languages called
flanked automata. Intuitively, a FFA $(\A, F)$ is a simple extension
of NFA where we add in the relation $F$ extra information that can be
used to check (non-deterministically) whether a word is not accepted
by \A. Hence a FFA can be used both to test whether a word is in the
language associated to \A or in its complement. As a consequence, we
obtain good complexity results for several interesting problems:
universality, language inclusion, \dots This idea of adding
extra-information to encode both a language and its complement seems
to be new. It is also quite different from existing approaches used to
to define subclasses of NFA with good complexity properties, like for
example unambiguity~\cite{schmidt78,stearns1985equivalence}.  Our work
could be extended in several ways.

First, we have implemented all our proposed algorithms and
constructions and have found that---for several examples coming from
the system verification domain---it was often easy to define a
flanking function for a given NFA (even though we showed in
Sect.~\ref{sec:testing-if-an} that it is not always possible). More
experimental work is still needed, and in particular the definition of
a good set of benchmarks.

Next, we have used the powerset construction multiple time in our
definitions. Most particularly as a way to test if a FFA is flanked
or if a NFA is flankable. Other constructions used to check language
inclusion or simulation between NFA could be useful in this context
like, for example, the antichain-based method~\cite{abdulla10}.

Finally, we still do not know how to compute a ``succinct'' flanked
automaton from a NFA that is not flankable. At the moment, our only
solution is to compute a minimal equivalent DFA (since DFA are always
flankable). While it could be possible to subsequently simplify the
DFA---which is known to be computationally hard~\cite{tao}, even
without taking into account the flanked function---it would be
interesting to have a more direct construction. This interesting open
problem is left for future investigations.

\subparagraph*{Acknowledgments} We thank Denis Kuperberg, Thomas
Colcombet, and Jean-Eric Pin for providing their expertise and insight
and for suggesting the example that led to the proof of
Theorem~\ref{thm:FFAvsDFA}.

\newpage

\bibliographystyle{plain}
\bibliography{ffa}

\begin{thebibliography}{10}

\bibitem{abdulla10}
Parosh~Aziz Abdulla, Yu-Fang Chen, Lukas Holik, Richard Mayr, and Tomas Vojnar.
\newblock When simulation meets antichains.
\newblock In {\em Tools and Algorithms for the Construction and Analysis of
  Systems}, volume 6015 of {\em LNCS}. Springer, 2010.

\bibitem{bauer12}
Sebastian~S. Bauer, Alexandre David, Rolf Hennicker, Kim Guldstrand~Larsen,
  Axel Legay, Ulrik Nyman, and Andrzej Wasowski.
\newblock Moving from specifications to contracts in component-based design.
\newblock In {\em Fundamental Approaches to Software Engineering}, volume 7212
  of {\em LNCS}, pages 43--58. Springer, 2012.

\bibitem{benveniste2008multiple}
Albert Benveniste, Beno{\^\i}t Caillaud, Alberto Ferrari, Leonardo Mangeruca,
  Roberto Passerone, and Christos Sofronis.
\newblock Multiple viewpoint contract-based specification and design.
\newblock In {\em Formal Methods for Components and Objects}, volume 5382 of
  {\em LNCS}, pages 200--225. Springer, 2008.

\bibitem{colcombet2012}
Thomas Colcombet.
\newblock {Forms of Determinism for Automata}.
\newblock In {\em 29th International Symposium on Theoretical Aspects of
  Computer Science (STACS 2012)}, volume~14, pages 1--23, 2012.

\bibitem{colcombet15:_flank}
Thomas Colcombet.
\newblock Flankable automata may be exponentially more succint than
  deterministic one.
\newblock private communication, March 2015.

\bibitem{tao}
Tao Jiang and B.~Ravikumar.
\newblock Minimal {NFA} problems are hard.
\newblock {\em SIAM Journal on Computing}, 22(6):1117--1141, 1993.

\bibitem{kao09:_nfas}
Jui-Yi Kao, Narad Rampersad, and Jeffrey Shallit.
\newblock On {NFAs} where all states are final, initial, or both.
\newblock {\em Theoretical Computer Science}, 410(47–49):5010–--5021, 2009.

\bibitem{Raclet200893}
Jean-Baptiste Raclet.
\newblock Residual for component specifications.
\newblock {\em Electronic Notes in Theoretical Computer Science}, 215:93--110,
  2008.
\newblock Proceedings of the 4th International Workshop on Formal Aspects of
  Component Software (FACS 2007).

\bibitem{schmidt78}
E.~M. Schmidt.
\newblock {\em Succinctness of Description of Context-Free, Regular and
  Unambiguous Languages}.
\newblock PhD thesis, Cornell University, 1978.

\bibitem{stearns1985equivalence}
Richard~Edwin Stearns and Harry~B Hunt~III.
\newblock On the equivalence and containment problems for unambiguous regular
  expressions, regular grammars and finite automata.
\newblock {\em SIAM Journal on Computing}, 14(3):598--611, 1985.

\bibitem{7202840}
T.~Villa, A.~Petrenko, N.~Yevtushenko, A.~Mishchenko, and R.~Brayton.
\newblock Component-based design by solving language equations.
\newblock {\em Proceedings of the IEEE}, PP(99):1--16, 2015.

\end{thebibliography}

\end{document}